\documentclass{article}
\usepackage{spconf}

\ninept
\usepackage{graphicx,verbatim}
\usepackage{amsmath}
\usepackage{amssymb}
\usepackage{float}

\newcommand{\cA}{{\mathcal A}}

\newcommand{\bF}{{\boldsymbol F}}
\newcommand{\ba}{{\boldsymbol a}}
\newcommand{\bQ}{{\boldsymbol Q}}
\newcommand{\bA}{{\boldsymbol A}}

\newcommand{\bTheta}{{\boldsymbol \Theta}}
\newcommand{\btheta}{{\boldsymbol \theta}}
\newcommand{\bb}{{\boldsymbol b}}
\newcommand{\br}{{\boldsymbol r}}
\newcommand{\bB}{{\boldsymbol B}}
 
\newcommand{\bC}{{\boldsymbol C}}
\newcommand{\bD}{{\boldsymbol D}}
\newcommand{\bE}{{\boldsymbol E}}
\newcommand{\bZ}{{\boldsymbol Z}}
\newcommand{\bc}{{\boldsymbol c}}

\newcommand{\bx}{{\boldsymbol x}}

\newcommand{\cT}{{\mathcal T}}

\newcommand{\bu}{{\boldsymbol u}}

\newcommand{\bI}{{\boldsymbol I}}
\newcommand{\bY}{{\boldsymbol Y}}
\newcommand{\bV}{{\boldsymbol V}}
\newcommand{\bX}{{\boldsymbol X}}

\newcommand{\bW}{{\boldsymbol W}}

 \DeclareMathOperator{\trace}{Tr}
 \DeclareMathOperator{\toep}{toep}
 \DeclareMathOperator{\diag}{diag}
  \DeclareMathOperator{\rank}{rank}
    \DeclareMathOperator{\argmin}{argmin}
  
\usepackage{cite}\usepackage{amsthm}\usepackage{dsfont}\usepackage{array}\usepackage{mathrsfs}\usepackage{comment}

\begin{document}
\theoremstyle{plain}\newtheorem{lemma}{\textbf{Lemma}}\newtheorem{theorem}{\textbf{Theorem}}\newtheorem{corollary}{\textbf{Corollary}}\newtheorem{assumption}{\textbf{Assumption}}\newtheorem{example}{\textbf{Example}}\newtheorem{definition}{\textbf{Definition}}
\newtheorem{prop}{\textbf{Proposition}}
\theoremstyle{definition}

\theoremstyle{remark}\newtheorem{remark}{\textbf{Remark}}

\title{Joint Sparsity Recovery for Spectral Compressed Sensing}
\name{Yuejie Chi\thanks{This work was partially supported by the Simons Foundation under grant No. 276631.}}
\address{Department of Electrical and Computer Engineering \\ The Ohio State University, Columbus, Ohio 43210}

\maketitle

\begin{abstract}
Compressed Sensing (CS) is an effective approach to reduce the required number of samples for reconstructing a sparse signal in an a priori basis, but may suffer severely from the issue of basis mismatch. In this paper we study the problem of simultaneously recovering multiple spectrally-sparse signals that are supported on the same frequencies lying arbitrarily on the unit circle. We propose an atomic norm minimization problem, which can be regarded as a continuous counterpart of the discrete CS formulation and be solved efficiently via semidefinite programming. Through numerical experiments, we show that the number of samples per signal may be further reduced by harnessing the joint sparsity pattern of multiple signals. 

\end{abstract}


\section{Introduction}

Compressed Sensing (CS) \cite{CandRomTao06,Don2006} asserts a signal can be recovered from a small number of linear measurements if it can be regarded as a sparse or compressible signal in an a priori basis. An important consequence is its application in analog-to-digital conversion and spectrum estimation, where the signal of interest is spectrally sparse, composing of a small number $r$ of frequency components. It is shown that if the frequencies of a signal all lie on the DFT grid, the signal of length $n$ can then be recovered exactly from a random subset of $O(r\log n)$ samples with high probability \cite{candes2007sparsity}. However, in reality the frequencies of a signal never lie on a grid, no matter how fine the grid is; rather, they are continuous-valued and determined by the mother nature, e.g. a point spread function. Performance degeneration of CS algorithms is observed and studied in \cite{chi2011sensitivity,scharf2011sensitivity,pakroohanalysis} when a ``basis mismatch'' between the actual frequencies and the assumed grid occurs, and many algorithms have been proposed to mitigate the basis mismatch effect \cite{Yang2011offgrid,fannjiang2012coherence}.

More recently, several approaches have been proposed to assume away the need for an a priori basis while maintaining the capabilities of subsampling in CS with rigorous performance guarantees. One recent approach is via atomic norm minimization \cite{chandrasekaran2012convex}, which is a general recipe for designing convex solutions to parsimonious model selection. It has been successfully applied to recover a spectrally-sparse signal from a subset of its consecutive samples \cite{CandesFernandez2012SR} or randomly selected samples \cite{TangBhaskarShahRecht2012}. In particular, Tang et. al. showed that a spectrally-sparse signal can be recovered from $O(r\log n \log r)$ random samples with high probability when the frequencies satisfy a mild separation condition \cite{TangBhaskarShahRecht2012}. This framework has also been extended to handle higher-dimensional frequencies \cite{chi2013compressive}. Another approach is proposed in \cite{chen2013spectral,chen2013robust}, where the problem is reformulated into a structured matrix completion inspired by matrix pencil \cite{Hua1992}. For this approach, it is shown that $O(r\log^2 n)$ randomly selected samples are sufficient to guarantee perfect recovery with high probability under some mild incoherence conditions and the approach is amenable to higher-dimensional frequencies.  With slightly more samples, both approaches can recover off-the-grid frequencies at an arbitrary precision. We refer interested readers for respective papers for details. 

In this paper, we study the problem of simultaneously recovering multiple spectrally-sparse signals that are supported on the same frequencies lying arbitrarily on the unit circle. With Multiple measurement Vectors (MMV) in a CS framework, it is shown to further reduce the required number of samples \cite{tropp2006algorithms2,tropp2006algorithms,lee2012subspace,kim2012compressive,davies2012rank,mishali2008reduce} by harnessing the joint sparsity pattern of different signals. Our proposed algorithm can be regarded as a continuous counterpart of the MMV model in CS, and is based on atomic norm minimization which can be solved efficiently using semidefinite programming.  We characterize a dual certificate for the optimality of the proposed algorithm. Comparisons are given for joint sparse recovery between conventional CS algorithms and the proposed algorithm through numerical experiments, which indicate that the number of samples per signal can be reduced by harnessing the joint sparsity pattern of multiple signals using atomic norm minimization. 

The rest of the paper is organized as below. Section~\ref{backgrounds} describes the problem formulation. Section~\ref{proposal} formulates the atomic norm minimization problem for joint sparse recovery. Section~\ref{numerical} gives numerical experiments and we conclude in Section~\ref{concl}. Throughout the paper, we use bold letters to denote matrices and vectors, and unbolded letters to denote scalars. The transpose is denoted by $(\cdot)^T$, the complex conjugate is denoted by $(\cdot)^*$, and the trace is denoted by $\trace(\cdot)$.
   
\section{Problem Formulation and Backgrounds} \label{backgrounds}
Let $\bx=[x_1,\ldots,x_n]^T\in\mathbb{C}^n$ be a spectrally-sparse signal with $r$ distinct frequency components, which can be written as
\begin{equation}\label{single}
\bx= \sum_{k=1}^r c_{k} \ba(f_k) \triangleq \bV  \bc,
\end{equation}
where each atom $\ba(f)$ is defined as
\begin{equation}
\ba(f ) = \frac{1}{\sqrt{n}}\left[1, e^{j2\pi f}, \ldots, e^{j2\pi f(n-1)}\right]^T , \quad f \in [0,1).
\end{equation}
The Vandermonde matrix $\bV$ is given as $\bV= [\ba(f_1),\ldots, \ba(f_r)]\in\mathbb{C}^{n\times r}$, and the coefficient vector $\bc=[c_1,\ldots,c_r]^T\in\mathbb{C}^r$. The set of frequencies $\mathcal{F}=\{f_k\}_{k=1}^r$ can lie anywhere on the unit circle, such that $f_k$ is continuously valued in $[0,1)$. 

In an MMV model, we consider $L$ signals, stacked in a matrix, $\bX=[\bx_1,\ldots,\bx_L]$, where each $\bx_l\in\mathbb{C}^{n}$, $l=1,\ldots, L$, shares the same set of frequencies and is composed of the form of \eqref{single} as
\begin{equation}
\bx_l = \sum_{k=1}^r c_{k,l} \ba(f_k) = \bV \bc_l,
\end{equation}
with $\bc_l = [c_{1,l},\ldots,c_{r,l}]^T$. Hence $\bX$ can be expressed as
\begin{equation}\label{data_rep}
 \bX = \bV\bC,
 \end{equation}
where $ \bC = \begin{bmatrix}
\bc_1 & \cdots & \bc_L 
\end{bmatrix}  \in\mathbb{C}^{r \times L}$. Assume we observe the \textit{same} subset of entries of each $\bx_l$, and  denote the location set of observed entries by $\Omega$, with $m=|\Omega|$. When there is no noise, the observations can be given as 
$$\bZ_{\Omega} = \mathcal{P}_{\Omega}(\bX) \in \mathbb{C}^{m\times L},$$ 
where $\mathcal{P}_{\Omega}$ is a projection operator that only preserves the rows of $\bX$ indexed by $\Omega$.\footnote{We remark that, this restriction of fixing the observation pattern across different signals, is actually unnecessary for the proposed algorithm.}

The traditional CS approach assumes that $\bx_l$'s are sparse in an a priori determined DFT basis or DFT frame $\bF\in\mathbb{C}^{n \times d}$ $(d\geq n)$, and they share the same sparsity pattern. Hence, the signal $\bX$ is \textit{modeled} as
\begin{equation} \label{model}
\bX = \bF \bTheta,
\end{equation} 
where $\bTheta=[\btheta_1,\ldots,\btheta_d]^*\in\mathbb{C}^{d\times L}$, and the number of nonzero rows of $\bTheta$ is small. Define the group sparsity $\ell_1$ norm of $\bTheta$ as $\|\bTheta\|_{2,1} =\sum_{i=1}^d \| \btheta_i\|_2$. A convex optimization algorithm \cite{tropp2006algorithms} that motivates group sparsity can be posed to solve the MMV model as
\begin{equation} \label{cs_mmv}
\hat{\bTheta}=\argmin_{\bTheta}\; \|\bTheta\|_{2,1} \quad \mbox{s.t.}\quad \bF_{\Omega}\bTheta = \bZ_{\Omega},
\end{equation}
where $\bF_{\Omega}$ is the subsampled DFT basis or DFT frame on $\Omega$. The signal then is recovered as $\hat{\bX}=\bF \hat{\bTheta}$. However, when the frequencies $\mathcal{F}$ are off-the-grid, the model \eqref{model} becomes highly inaccurate due to spectral leakage along the Dirichlet kernel, making \eqref{cs_mmv} degrades significantly in performance. We will compare against this conventional approach with our proposed algorithm in Section~\ref{numerical}.

\section{Atomic Norm Minimization For MMV Models} \label{proposal}
 In this section we develop the atomic norm minimization algorithm for solving the MMV model with spectrally-sparse signals. We first define an atom as
 \begin{equation}
\bA(f,\bb ) = \ba(f)\bb^* ,
\end{equation}
where $f \in [0,1)$, $\bb\in\mathbb{R}^{L}$ satisfying $\|\bb\|_2 =1$, and the set of atoms as $\cA=\{\bA(f,\bb ) | f \in [0,1), \|\bb\|_2 =1\}$. Define
\begin{align*}
\| \bX\|_{\cA,0} &  = \inf_r \left\{   \bX = \sum_{k=1}^r c_k \bA(f_k,\bb_k), c_k \geq 0 \right\}.
\end{align*}
as the smallest number of atoms to describe $\bX$. Our goal is thus to minimize $\| \bX\|_{\cA,0}$ that satisfies the observation, given as
\begin{equation} \label{atomic_zero}
\min \|\bX\|_{\cA,0} \quad \mbox{s.t.} \quad \bZ_\Omega = \mathcal{P}_{\Omega}(\bX) .
\end{equation}
It is easy to show that $\| \bX\|_{\cA,0}$ can be represented equivalently as
\begin{align*}
\| \bX\|_{\cA,0} &  =  \inf_{\bu ,\bW } \left\{  \rank(\toep(\bu))  \Big| \begin{bmatrix}
\toep(\bu) & \bX \\
\bX^* & \bW \end{bmatrix} \succ \bf{0} \right\},
\end{align*}
where $\toep(\bu)$ is the Toeplitz matrix generated by $\bu$. Hence \eqref{atomic_zero} is NP-hard. We will alternatively consider the convex relaxation of $\| \bX\|_{\cA,0}$, defining the atomic norm \cite{chandrasekaran2012convex} of $\bX$ as
\begin{align}
\| \bX\|_{\cA} &= \inf \left\{ t>0: \; \bX\in t\; \mbox{conv}(\cA) \right\} \nonumber \\
&= \inf \left\{ \sum_k c_k \Big| \bX = \sum_k c_k \bA(f_k,\bb_k), c_k \geq 0 \right\}. \label{atomic_def}
\end{align}
The atomic norm of a single vector $\bx_l$ defined in \cite{TangBhaskarShahRecht2012} becomes a special case of \eqref{atomic_def} for $L=1$. We propose to solve the following atomic norm minimization algorithm:
\begin{equation}\label{primal}
\hat{\bX}=\argmin_\bX \|\bX\|_\cA \quad \mbox{s.t.} \quad  \bZ_\Omega = \mathcal{P}_{\Omega}(\bX).
\end{equation}


\subsection{Semidefinite Program Characterization}

We now prove the following equivalent semidefinite program (SDP) characterization of $\|\bX\|_\cA$, showing that \eqref{primal} can be solved efficiently using off-the-shelf SDP solvers.

\begin{theorem}\label{atomic-sdp} The atomic norm $\|\bX\|_\cA$ can be written equivalently as
\begin{align*}
 \| \bX\|_\cA = \inf_{\bu\in\mathbb{C}^n,\bW\in\mathbb{C}^{L\times L}} & \Big\{ \frac{1}{2}\trace(\toep(\bu)) + \frac{1}{2}\trace(\bW) \Big|  \\
 & \begin{bmatrix}
\toep(\bu) & \bX \\
\bX^* & \bW \end{bmatrix} \succ \bf{0} \Big\}.
\end{align*}
\end{theorem}

\begin{proof} Denote the value of the right hand side as $\|\bX\|_\cT$. Suppose that $\bX =\sum_{k=1}^r c_k \ba(f_k)\bb_k^*$, there exists a vector $\bu$ such that
$$ \toep(\bu) = \sum_{k=1}^r c_k \ba(f_k) \ba(f_k)^*, $$
by the Vandermonde decomposition lemma \cite{Caratheodory}. It is obvious that
\begin{align*} 
&\quad \begin{bmatrix}
\toep(\bu) & \bX \\
\bX^* & \sum_{k=1}^r c_k \bb_k \bb_k^* \end{bmatrix} \\
&=  \begin{bmatrix}
\sum_{k=1}^r c_k \ba(f_k) \ba(f_k)^* & \sum_{k=1}^r c_k\ba(f_k) \bb_k^* \\
 \sum_{k=1}^r  c_k \bb_k\ba(f_k)^* & \sum_{k=1}^r c_k \bb_k \bb_k^* \end{bmatrix} \\
 & = \sum_{k=1}^r  c_k  \begin{bmatrix}
  \ba(f_k)  \\
  \bb_k    \end{bmatrix} \begin{bmatrix}
  \ba(f_k)^* &  \bb_k^*    \end{bmatrix} \succ 0,
 \end{align*}
 and 
$$ \frac{1}{2}\trace(\toep(\bu)) + \frac{1}{2}\trace(\bW) = \sum_{k=1}^r c_k = \| \bX\|_\cA,$$
therefore $\|\bX\|_{\cT}\leq \|\bX\|_\cA$. On the other hand, suppose that for any $\bu$ and $\bW$ that satisfy
$$ \begin{bmatrix}
\toep(\bu) & \bX \\
\bX^* & \bW \end{bmatrix} \succ \bf{0}, $$
with $\toep(\bu) = \bV\bD\bV^*$, $\bD=\diag(d_i)$, $d_i> 0$. It follows that $\bX$ is in the range of $\bV$, hence $\bX=\bV\bB$ with the columns of $\bB^T$ given by $\bb_i$. Since $\bW$ is also PSD, $\bW$ can be written as $\bW = \bB^*\bE\bB$ where $\bE$ is also PSD. We now have
$$ \begin{bmatrix}
\toep(\bu) & \bX \\
\bX^* & \bW \end{bmatrix} =\begin{bmatrix}
\bV &   \\
  & \bB^* \end{bmatrix} \begin{bmatrix}
\bD & \bI \\
\bI & \bE \end{bmatrix} \begin{bmatrix}
\bV^* &  \\
 & \bB \end{bmatrix} \succ \bf{0}, $$
which yields 
$$ \begin{bmatrix}
\bD & \bI \\
\bI & \bE \end{bmatrix} \succ \bf{0} $$
and $\bE\succ \bD^{-1}$ by the Schur complement lemma. Now observe
\begin{align*}
\trace(\bW)& = \trace(\bB^*\bE\bB)\geq  \trace(\bB^*\bD^{-1}\bB) \\
&= \trace(\bD^{-1}\bB\bB^*) = \sum_i d_i^{-1} \| \bb_i\|^2.
\end{align*}
Therefore,
\begin{align*}
&\quad\frac{1}{2}\trace(\toep(\bu))+\frac{1}{2} \trace(\bW)  =\frac{1}{2} \trace(\bD)+\frac{1}{2} \trace(\bW) \\
& \geq \sqrt{ \trace(\bD)\cdot \trace(\bW)} \\
& \geq \sqrt{\left(\sum_i d_i \right) \left( \sum_i d_i^{-1} \| \bb_i\|^2 \right) }  \\
& \geq \sum \|\bb_i\| \geq \|\bX \|_\cA,
\end{align*}
which gives $\|\bX\|_\cT \geq \|\bX\|_\cA$. Therefore, $\|\bX\|_\cT = \|\bX\|_\cA$.
\end{proof}
From Theorem~\ref{atomic-sdp}, we can now equivalently write \eqref{primal} as the following SDP:
\begin{align}
\hat{\bX}&= \argmin_{\bX}\inf_{\bu,\bW} \;  \frac{1}{2}\trace(\toep(\bu)) + \frac{1}{2}\trace(\bW) \label{primal-sdp}\\
& \mbox{s.t.} \;  \begin{bmatrix}
\toep(\bu) & \bX \\
\bX^* & \bW \end{bmatrix} \succ \mathbf{0} , \; \bZ_{\Omega} = \mathcal{P}_{\Omega}(\bX). \nonumber
\end{align}

\subsection{Dual Certification}
An important question is when the algorithm \eqref{primal} admits perfect recovery. To this end, we study the dual problem of \eqref{primal}. Denote the optimal solution of \eqref{primal} as $\bX^{\star}$. Let $\bY\in\mathbb{C}^{n\times L}$, define $\langle \bY,\bX \rangle = \trace(\bX^*\bY)$, and $\langle \bY,\bX \rangle_{\mathbb{R}}=\mbox{Re}(\langle \bY,\bX \rangle)$. The dual norm of $\|\bX\|_\cA$ can be defined as
\begin{align*}
\|\bY\|_{\cA}^* &=\sup_{\|\bX\|_\cA\leq 1} \langle \bY,\bX\rangle_{\mathbb{R}} \\
&  = \sup_{f\in[0,1),\|\bb\|_2=1}  \langle \bY, \ba(f) \bb^* \rangle_{\mathbb{R}} \\
&  = \sup_{f\in[0,1),\|\bb\|_2=1} \left| \langle \bb, \bY^*\ba(f)   \rangle \right| \\
& = \sup_{f\in[0,1)} \|\bY^*\ba(f)  \|_2 \triangleq \sup_{f\in[0,1)} \| \bQ(f)\|_2,
\end{align*}
where $\bQ(f) = \bY^*\ba(f)$ is a length-$L$ vector with each entry a polynomial in $f$. The dual problem of \eqref{primal} can thus be written as
\begin{equation} \label{dual}
\max_{\bY} \; \langle \bY_\Omega, \bX_{\Omega}^{\star} \rangle_{\mathbb{R}} \quad \mbox{s.t.} \quad \|\bY\|_{\cA}^* \leq 1, \bY_{\Omega^c} = 0,
\end{equation}
where $\bY_{\Omega^c}$ denotes the projection of $\bY$ on the rows denoted by the location set $\Omega^c=\{1,\ldots, n\} \backslash \Omega$.
Let $(\bX,\bY)$ be primal-dual feasible to \eqref{primal} and \eqref{dual}, we have $\langle \bY,\bX\rangle_{\mathbb{R}}=\langle \bY,\bX^{\star}\rangle_{\mathbb{R}}$. Strong duality holds since Slater's condition holds, and it implies that the solutions of \eqref{primal} and \eqref{dual} equal if and only if $\bY$ is dual optimal and $\bX$ is primal optimal \cite{boyd2004convex}. Using strong duality we can obtain a dual certification to the optimality of the solution of \eqref{primal}.

\begin{prop} \label{dual_certificate}The solution of \eqref{primal} $\hat{\bX}=\bX^{\star}$ is the unique optimizer if there exists $\bQ(f)= \bY^*\ba(f) $ 
such that
\begin{equation} \label{conditions}
\begin{cases}
(C1):\quad \bQ(f_k) = \bb_k, & \forall f_k \in \mathcal{F}, \\
(C2):\quad \| \bQ(f) \|_2 < 1, &\forall f\notin \mathcal{F}, \\
(C3):\quad \bY_{\Omega^c} = 0.&
\end{cases}
\end{equation}

\end{prop}

\begin{proof} First, any $\bY$ satisfying \eqref{conditions} is dual feasible. We have
\begin{align*}
\|\bX^{\star}\|_{\cA} & \geq\|\bX^{\star}\|_{\cA}\|\bY\|_{\cA}^* \\
&\geq \langle \bY, \bX^{\star} \rangle_{\mathbb{R}}  =\langle \bY,  \sum_{k=1}^r c_k \ba(f_k)\bb_k^* \rangle_{\mathbb{R}} \\
& = \sum_{k=1}^r \mbox{Re}\left( c_k \langle \bY, \ba(f_k)\bb_k^* \rangle \right) \\
& = \sum_{k=1}^r \mbox{Re} \left( c_k \langle \bb_k, \bQ(f_k) \rangle \right) \\
& =  \sum_{k=1}^r \mbox{Re} \left( c_k \langle \bb_k, \bb_k \rangle \right)  = \sum_{k=1}^r c_k \geq \|\bX^{\star}\|_{\cA} .
\end{align*}
Hence $\langle \bY, \bX^{\star} \rangle_{\mathbb{R}} =\|\bX^{\star}\|_{\cA}$. By strong duality we have $\bX^{\star}$ is primal optimal and $\bY$ is dual optimal.

For uniqueness, suppose $\hat{\bX}$ is another optimal solution which has support outside $\mathcal{F}$. It is trivial to justify if $\hat{\bX}$ and $\bX^{\star}$ have the same support, they must coincide since the set of atoms with frequencies in $\mathcal{F}$ is independent. Let $\hat{\bX}=\sum_k \hat{c}_k \ba(\hat{f}_k)\hat{\bb}_k^*$. We then have 
\begin{align*}
\langle\bY, \hat{\bX}  \rangle_{\mathbb{R}}  & = \sum_{\hat{f}_k\in\mathcal{F}}\mbox{Re}  \left( \hat{c}_k \langle \hat{\bb}_k, \bQ(\hat{f}_k) \rangle \right) +  \sum_{\hat{f}_l\notin\mathcal{F}}\mbox{Re}  \left( \hat{c}_l \langle \hat{\bb}_l, \bQ(\hat{f}_l) \rangle \right) \\
& <  \sum_{\hat{f}_k\in\mathcal{F}} \hat{c}_k + \sum_{\hat{f}_l \notin\mathcal{F}} \hat{c}_l  = \|\hat{\bX}\|_{\cA},
\end{align*}
which contradicts strong duality. Therefore the optimal solution of \eqref{primal} is unique.
\end{proof}
Proposition~\ref{dual_certificate} offers a way to certify the optimality of \eqref{primal} as long as we can find a dual polynomial $\bQ(f)$ that satisfies \eqref{conditions}. This is left for future work.

\begin{figure*}
\centering
\begin{tabular}{cc}
\includegraphics[width=0.44\textwidth]{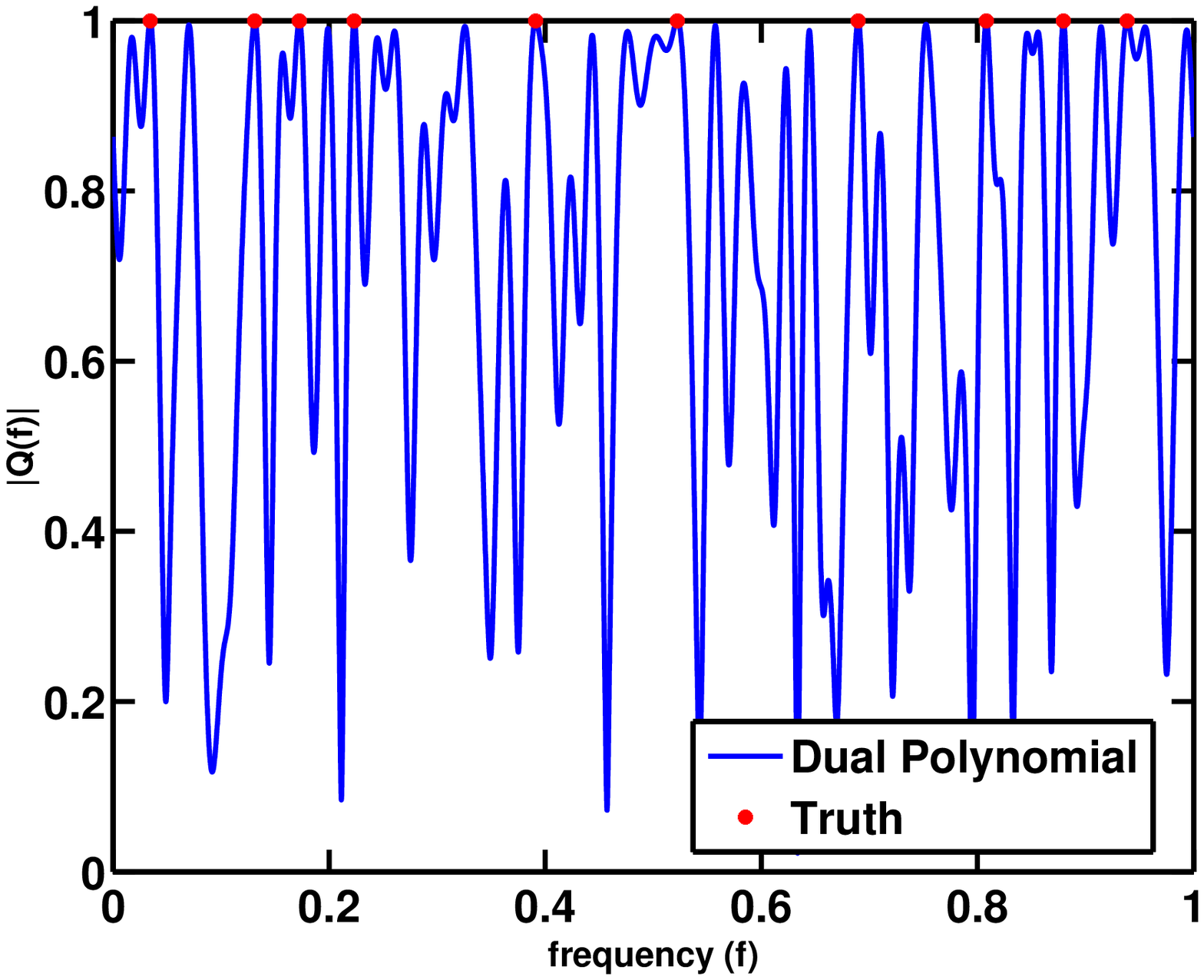} & \includegraphics[width=0.44\textwidth]{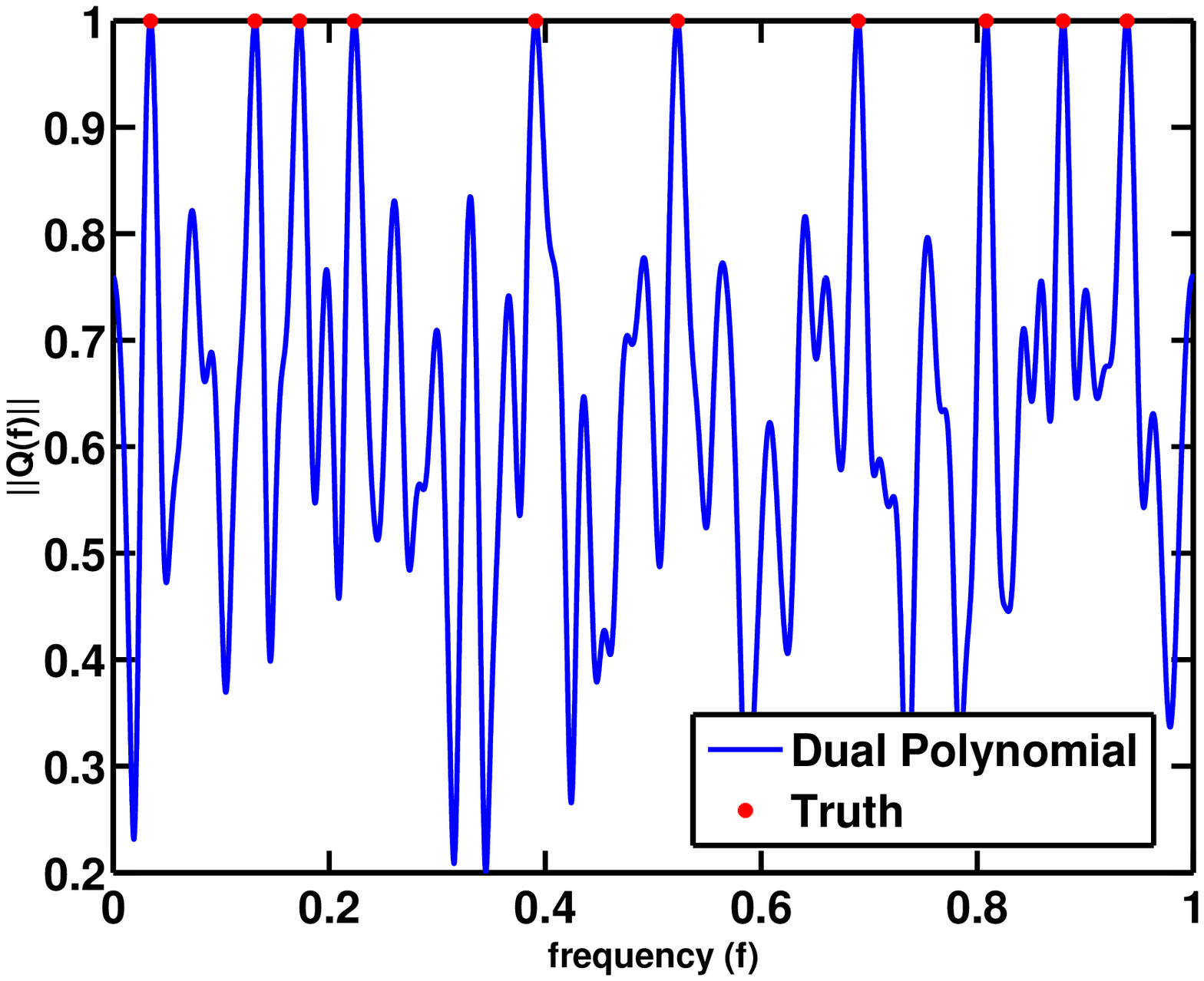} \\
(a) $  L=1$ &  (b) $  L=3$ 
\end{tabular}
\caption{The reconstructed dual polynomial for a randomly generated spectrally-sparse signal with $r=10$, $n=64$, and $m=32$: (a) $L=1$, (b) $L=3$.} \label{dualpoly}
\end{figure*}
\section{Numerical Experiments} \label{numerical}

In this section, we evaluation the performance of the proposed algorithm \eqref{primal}, showing that the number of measurements per signal may be reduced as we increase the number of signals $L$ for achieving the same performance.

\subsection{Phase transition when varying the number of signals}
Let $n=64$ and $m=32$. In each Monte Carlo experiment, we randomly generate a spectrally-sparse signal with $r$ frequencies randomly located in $[0,1)$ that satisfies a separation condition $\Delta =\min_{k\neq l} |f_k-f_l| \geq 1/n$. This separation condition is slightly weaker than the condition asserted in \cite{TangBhaskarShahRecht2012} to guarantee the success of \eqref{primal} with high probability for $L=1$. For each frequency component, we randomly generate the amplitude for each signal from the standard complex Gaussian distribution $\mathcal{CN}(0,1)$. We run \eqref{primal-sdp} using CVX \cite{grant2008cvx} and calculate the reconstruction normalized mean squared error (NMSE) as $\|\hat{\bX}-\bX^{\star} \|_F/\|\bX^{\star}\|_F$ where $\bX^{\star}$ is the ground truth. The experiment is claimed successful if $\mbox{NMSE}\leq 10^{-5}$. For each pair of $r$ and $L$, we run a total of $50$ Monte Carlo experiments and output the average success rate. Fig.~\ref{fixed_m} shows the success rate of reconstruction versus the sparsity level $r$ for $L=1$, $2$, and $3$ respectively. As we increase $L$, the success rate becomes higher for the same sparsity level.
\begin{figure}[h]
\centering
\includegraphics[width=0.45\textwidth]{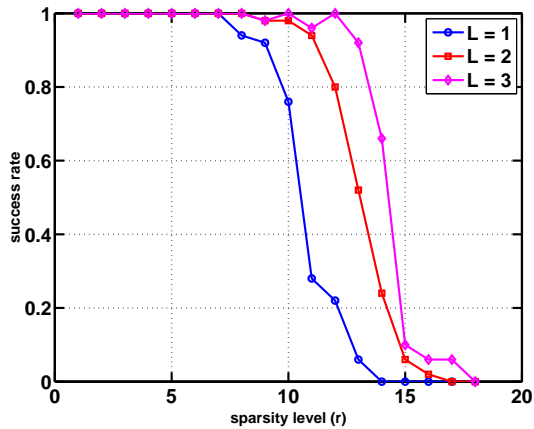}
\caption{Success rate of reconstruction versus the sparsity level $r$ for $L=1,2,3$ when $n=64$, $m=32$ and the frequencies are generated satisfying a separation condition $\Delta\geq1/n$.} \label{fixed_m}
\end{figure}

Fig.~\ref{dualpoly} shows the reconstructed dual polynomial for a randomly generated spectrally-sparse signal with $r=10$ when $L=1$ and $L=3$ respectively. It can be seen that although the algorithm achieves perfect recovery for both cases, the reconstructed dual polynomial has a much better localization property when multiple signals are present.

\begin{figure}[htp]
\centering
\begin{tabular}{cc}
\hspace{-0.1in}\includegraphics[height=1.8in, width=0.25\textwidth]{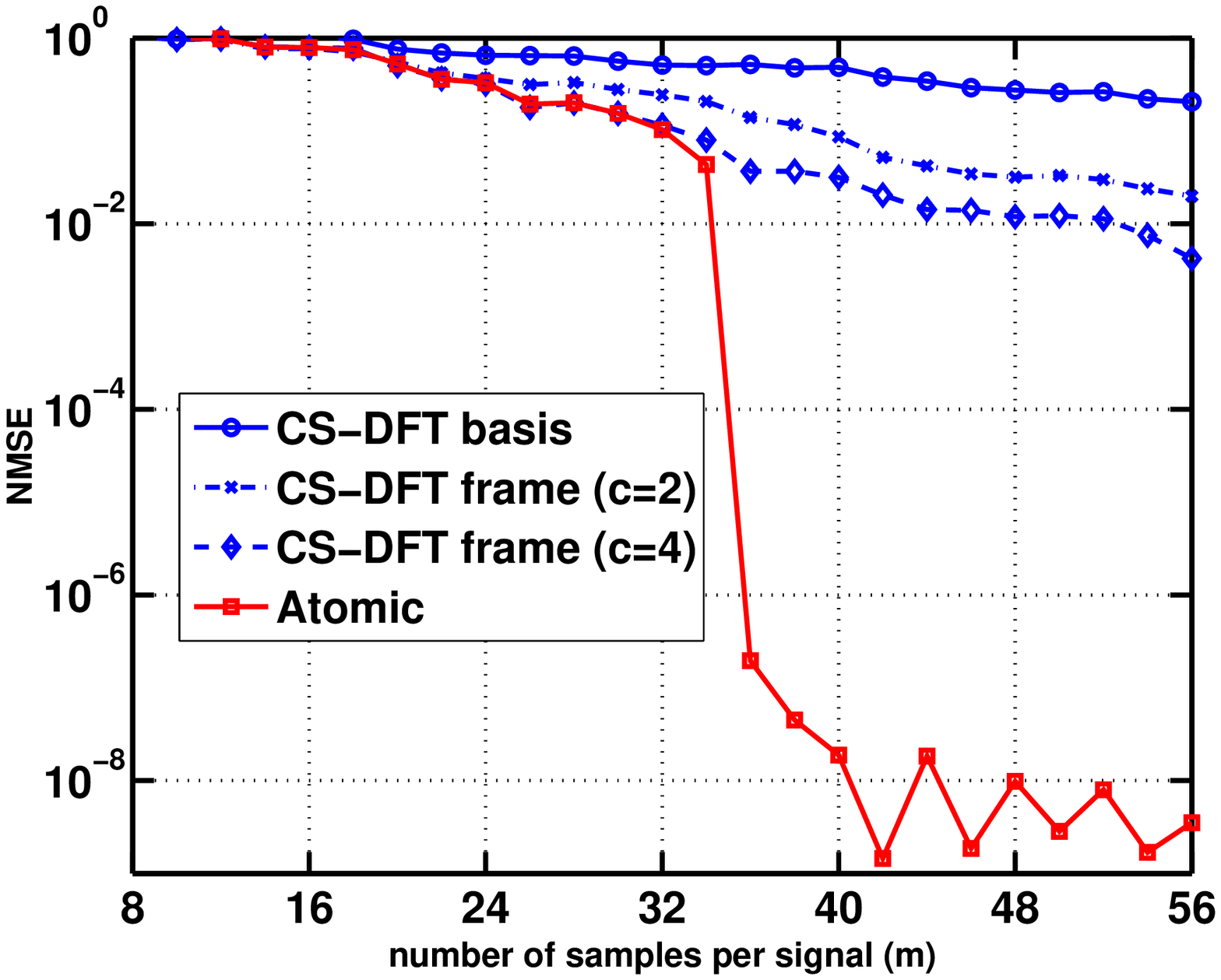} & \hspace{-0.15in}
\includegraphics[height=1.8in,width=0.25\textwidth]{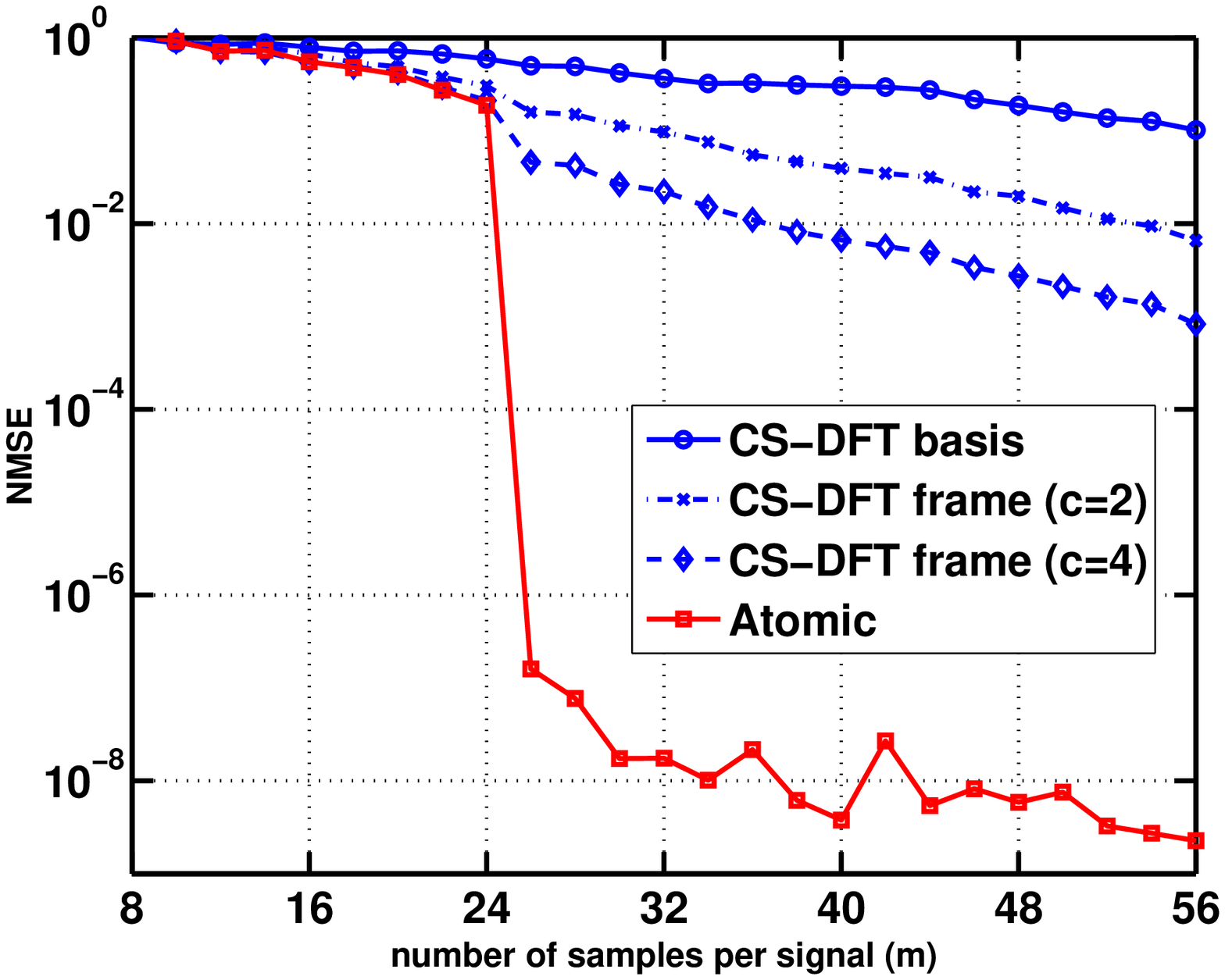} \\
(a) $L=1$ &   (b) $L=3$
\end{tabular}
\caption{The NMSE of reconstruction versus the number of samples per signal $m$ for a randomly generated spectrally-sparse signal satisfying a separation condition $\Delta\geq 1/n$ with $r=8$, $n=64$: (a) $L=1$, (b) $L=3$.} \label{cmp_offgrid}
\end{figure}
 \subsection{Comparison of CS and the proposed algorithm}
Let $n=64$. We randomly generate $r=10$ frequencies on the unit circle $[0,1)$ satisfying a separation condition $\Delta\geq 1/n$ as in the previous setup. Fig.~\ref{cmp_offgrid} shows the NMSE of the reconstructed signal using the CS-MMV algorithm \eqref{cs_mmv} with a DFT basis, a DFT frame with an oversampling factor $c=2$, a DFT frame with an oversampling factor $c=4$, and the proposed atomic norm minimization algorithm \eqref{primal} for (a) $L=1$ and (b) $L=3$. The atomic norm minimization algorithm outperforms \eqref{cs_mmv} even when the reconstruction is not exact at smaller values of $m$. The CS algorithm \eqref{cs_mmv} can never achieve exact recover since a randomly generated frequency is always off the grid, while the recovery of the atomic norm algorithm \eqref{primal} is exact after $m$ exceeds a certain threshold for success. As we increase the number of signals $L$, both CS algorithms and the atomic norm algorithm improve their performance.

\section{Conclusions} \label{concl}
In this paper we study the problem of simultaneously recovering multiple spectrally-sparse signals that are supported on the same frequencies lying arbitrarily on the unit circle. We propose an atomic norm minimization problem, and solve it efficiently via semidefinite programming. Through numerical experiments, we show that the number of samples per signal may be further reduced by harnessing the joint sparsity pattern of multiple signals. The proposed atomic norm minimization algorithm can also be applied when the observation patterns are different across different signals. Future work is to develop theoretical guarantees of the proposed algorithm and examine its performance in the presence of noise. Another interesting direction is to study MMV extensions of the off-the-grid approach in \cite{chen2013spectral,chen2013robust}.


\bibliographystyle{IEEEtran} 
\bibliography{bibfileSparseMatrixPencil_atomic}

\end{document}